\documentclass[aip,jmp,author-year,reprint]{revtex4-1}
\usepackage{amssymb, amsmath, latexsym}
\usepackage{amsthm}
\usepackage[T1]{fontenc}
\usepackage[english]{babel}
\usepackage{microtype}
\usepackage{svn}

\theoremstyle{plain}
\newtheorem{thm}{Theorem}
\newtheorem{prop}[thm]{Proposition}
\newtheorem{lem}[thm]{Lemma}
\newtheorem{cor}[thm]{Corollary}
\theoremstyle{definition}
\newtheorem{defn}[thm]{Definition}

\theoremstyle{remark}
\newtheorem{rmk}[thm]{Remark}
\newtheorem{exa}[thm]{Example}

\usepackage{slashed}

\declareslashed{}{/}{.05}{.1}{\nabla} 

\declareslashed{}{/}{0}{-.1}{\triangle}

\newcommand*\set[2]{\ensuremath{\left\{\; #1 \;\middle|\; #2 \;\right\}}}

\newcommand*\Real{\ensuremath{\mathbb{R}}}

\newcommand*\Rational{\ensuremath{\mathbb{Q}}}

\makeatletter
\newcommand*\D{\@ifstar\D@Star\D@noStar}
\newcommand*\D@noStar[1][]{\ensuremath{~\mathrm{d}#1}}
\newcommand*\D@Star[1][]{\ensuremath{\mathrm{d}#1}}
\makeatother

\newcommand*\eqdef{\ensuremath{\overset{\mbox{\tiny{def}}}{=}}}

\makeatletter

\newcommand*{\Rmnum}[1]{\expandafter\@slowromancap\romannumeral #1@}
\makeatother

\newcommand*\dirlim\varinjlim
\newcommand*\invlim\varprojlim

\SVN $Date: 2013-11-08 17:08:58 +0100 (Fri, 08 Nov 2013) $
\SVN $Rev: 346 $

\begin{document}
\title{A comment on the construction of the maximal globally
hyperbolic Cauchy development}
\author{Willie Wai-Yeung Wong}
\affiliation{\'{E}cole Polytechnique F\'{e}d\'{e}rale de Lausanne,
Switzerland}
\email{willie.wong@epfl.ch}

\begin{abstract}
Under mild assumptions, we remove all traces of the axiom of choice
from the construction of the maximal globally hyperbolic Cauchy
development in general relativity. The construction relies on the
notion of direct union manifolds, which we review. The construction
given is very general: any physical theory with a suitable geometric
representation (in particular all classical fields), and such that a
strong notion of ``local existence and uniqueness'' of solutions for
the corresponding initial value problem is available, is amenable to
the same treatment. 
\end{abstract}

\maketitle

A celebrated theorem on the local Cauchy problem for Einstein's
equations is that of \citet{ChoGer1969} 
which asserts that
every initial data set leads to a unique \emph{maximal globally
hyperbolic Cauchy development}. In their original proof
(as well as many subsequent treatments, see
e.g.~\cite{Ringst2009}) the authors appealed to Zorn's Lemma in their
construction of the space-time manifold, which led to the common
misconception that the proof is non-constructive as the argument
seemingly depends on the axiom of choice. 

In this paper, we will show that, insofar as the actual
\emph{construction} of the space-time manifold is concerned, the use
of axiom of choice is not necessary. As it turns out, however, the
manifold constructed will not, in general, be second countable, making
geometry and analysis somewhat awkward on the space-time. One can
circumvent this difficulty in two ways: firstly in many situations assuming
the axiom of \emph{countable choice} (or even weaker statements
\citep{HowRub1998} such
as ``every countable union of countable sets is countable'') can allow
us to recover statements about countability of a basis for the
topology; secondly, an option that the author hopes to emphasize here,
is that sometimes adding some \emph{additional structures} (in a
manner that is natural and physical) to the definition of a
space-time will allow us to sidestep the issue of choice  entirely. 

Recently the same question, in the context of general relativity, has been 
treated exhaustively in a pre-print by \citet{Sbi2013}. Our approach
here offers two minor improvements:
\begin{enumerate}
\item We were able to avoid the axiom of choice entirely,
\emph{including the axiom of countable choice}. In Sbierski's
construction he appealed to a theorem of
\citet[Appendix]{Geroch1968} to obtain second countability of the
space-time. This theorem depends on the statement ``every countable
union of countable sets is countable'' alluded to above, which cannot
be proven in $\mathsf{ZF}$ (that is, Zermelo-Fraenkel set theory
without axiom of choice) alone. 
\item We isolate the structures which allows the general construction
to proceed. In particular, we shall take as \emph{black boxes} certain
facts about the local existence and uniqueness theorems (see, e.g.
\citet{ChoGer1969}, the monograph of 
\citet{Ringst2009}, as well as \citet{Sbi2013} for a discussion), and
concentrate only at the level of the construction of the maximal
development. This allows us to easily ``swap out'' the underlying
physical model with anything else that satisfies sufficiently strong
local existence and uniqueness theorems. 
\end{enumerate}

One may ask why bother at all about the issue of choice (and its
weaker formulations): after all, much of the foundations of topology,
geometry, and analysis that come up in the study of partial
differential equations on manifolds (a subject within which the
evolution problem of general relativity squarely sits) as commonly
used depend on some (perhaps weakened) version of the axiom of choice.
Here is a sampling of the statements that one may find useful but
cannot be proved (in $\mathsf{ZF}$) without some form of
choice (the numbers in parentheses refer to the searchable 
form numbers from the companion
website \url{http://consequences.emich.edu/conseq.htm} to
\citet{HowRub1998}):
\begin{itemize}
\item In functional analysis: Hahn-Banach theorem (\#52), Krein-Milman theorem (\#65), the Banach-Alaoglu theorem
(\#14Q), and the Arzel\`a-Ascoli theorem (\#94Q).
\item In analysis of metric spaces: the fact that on a metric space
sequential continuity implies continuity (\#8E), the Heine-Borel
theorem for $\Real^n$ (\#74), and that every uncountable subset
of $\Real$ contains a condensation point (\#6A).
\item In topology: that a second countable topological space is
separable (\#8L) and Urysohn's Lemma (\#78). 
\end{itemize}
One answer to that question ``why'' is one of aesthetics. (For 
some other points-of-view, the
author encourages the reader to look at the MathOverflow discussion
accessible at
\url{http://mathoverflow.net/questions/22927/}.) An insistence on using some
versions of axiom of choice when it can be easily avoided seems rather
wasteful. Furthermore, the general tool that was
actually used to construct the maximum globally hyperbolic Cauchy
development can be easily adapted in many categories other than just
smooth manifolds. One need not be in a context
where some version of axiom of choice is already used or needed to find
the construction useful.  Lastly, should
(however unlikely) the prevailing opinion on the
axiom of choice change in the future, a better understanding of a
physical theory's dependence on said axiom will certainly help
reformulate the theory upon more culturally acceptable foundations.

\section{The direct union construction}\label{sec:directunion}
We are motivated by the following: let $M$ be a smooth manifold and
$\mathcal{U}$ a collection of \emph{open} submanifolds of $M$ (in
other words, open sets on $M$), their union $\cup \mathcal{U}$ is
another open submanifold of $M$. With great hindsight, we see that in
the case of general relativity, we can take $M$ to be the maximal
globally hyperbolic Cauchy development, and $\mathcal{U}$ the
collection of Cauchy developments. Then \emph{morally speaking} we
should be able to obtain the maximal development as the union of the
elements of $\mathcal{U}$. The concept which allows us to consider the
union of a family of objects which do not, \textit{a priori}, exist as
a subset of the same set is the notion of \emph{direct
union}. Here we give a brief review. 

Recall that a directed set $(I,\prec)$ is a preorder such that every
pair of elements have an upper-bound. In the sequel a smooth manifold
refers to a topological manifold equipped with a $C^\infty$ structure.
We do \emph{not} assume the smooth manifold to be either Hausdorff or
second countable. 
\begin{defn}
Let $(I,\prec)$ be a directed set. A \emph{direct system of smooth
manifolds} over $(I,\prec)$ is the pair $(\mathfrak{M},\mathfrak{F})$
where $\mathfrak{M} = \{M_i\}_{i\in I}$ is a set of smooth
manifolds, and 
\[ \mathfrak{F} = \{f_{ji}\}_{i\prec j; i,j\in I}~,\qquad
\text{where}~ f_{ji}: M_i
\to M_j \]
is smooth, satisfying the condition that whenever $i\prec j \prec k$, 
\[ f_{ki} = f_{kj}\circ f_{ji}~.\]
\end{defn}

\begin{defn}
A direct system of smooth manifolds $(\mathfrak{M},\mathfrak{F})$ is
said to be \emph{regular} if all the maps $f_{ji}$ are open, and are
diffeomorphisms of $M_i$ onto their image. 
\end{defn}

Now, suppose $M,N$ are smooth manifolds, and $f:M\to N$ is an open
smooth map and a diffeomorphism onto $f(M)$. If $(U,\phi)$ is a 
chart of $M$, then by definition $(f(U),\phi\circ f^{-1}|_{f(U)})$ is 
a compatible chart with the atlas of $N$. This is to say.

\begin{lem}
Let $(\mathfrak{M},\mathfrak{F})$ be a regular direct system of smooth
manifolds. Let $\mathcal{A}_i$ be the maximal atlas (see Remark
\ref{rmk:atlas} below) for $M_i$, then the 
\emph{pushforward} 
$f_{ji}(\mathcal{A}_i)$ is
well-defined, and $f_{ji}(\mathcal{A}_i) \subseteq \mathcal{A}_j$. 
\end{lem}

\begin{rmk}\label{rmk:atlas}
In the spirit of the present paper, we remark that
the \emph{maximal atlas} is just the union over the set of all 
atlases compatible
to a given one, and its existence does not require Zorn's Lemma; there 
seems to be a lot of confusion in the literature
regarding this point (see for example \citet{Schwar2011,Mirand1995};
compare with the simpler treatment on pg.2 of \citet{KobNom1996}).
\end{rmk}

Now let us be given such a regular direct system of smooth manifolds.
We denote by $\dirlim_{\text{Top}} \mathfrak{M}$ its direct limit as a
topological space. That is to say, as a set we take
\[ \dirlim_{\text{Set}} \mathfrak{M} = \coprod_i M_i / \sim \]
where the equivalence relation is $x_i \sim x_j$ iff there exists $k
\succ i,j$ such that $f_{ki}(x_i) = f_{kj}(x_j)$;  we then give it
the quotient topology induced from $\coprod_i M_i$. We let $f_{*i}:M_i
\to \dirlim_{\text{Top}} \mathfrak{M}$ the natural mapping, which we
remark has the property that if $i\prec j$, $f_{*i} = f_{*j}\circ
f_{ji}$. By the definition of the
quotient topology we have that $f_{*i}$ is continuous. From the
definition of the equivalence relation, and the assumption that
$f_{ji}$ are injective, we also have that $f_{*i}$ is injective. 
We claim that $f_{*i}$ is also open. Indeed, let $U_i\subseteq M_i$ 
be an open set.
If $j\in I$ then since it is directed there exists $k \succ i,j$. By
definition $f_{*j}^{-1}\circ f_{*i}(U_i) = f_{kj}^{-1}\circ
f_{ki}(U_i)$, which is open since $f_{ki}$ is open and $f_{kj}$
is continuous. 

To finish the construction we need to give a smooth structure to
$\dirlim_{\text{Top}} \mathfrak{M}$. Consider the charts
$(f_{*i}(U_i), \phi \circ f_{*i}^{-1})$ (which is well-defined since
$f_{*i}$ is injective) where $(U_i,\phi) \in \mathcal{A}_i$. Clearly
the collection of all such charts cover $\dirlim_{\text{Top}}
\mathfrak{M}$. It suffices to show that they are pairwise compatible.
But if $(U_i,\phi)\in \mathcal{A}_i$ and $(V_j,\psi)\in \mathcal{A}_j$
are two charts, by assumption we can find $k\succ i,j$ such that
$(f_{ki}(U_i),\phi\circ f_{ki}^{-1}), (f_{kj}(V_j), \psi\circ
f_{kj}^{-1}) \in \mathcal{A}_k$. Using now that $f_{*k}$ is injective
and open, we conclude that $(f_{*i}(U_i),\phi\circ f_{*i}^{-1})$ and
$(f_{*j}(V_j),\psi\circ f_{*j}^{-1})$ are compatible. 

\begin{defn}
We denote by $\dirlim \mathfrak{M}$ the  
topological space $\dirlim_{\text{Top}} \mathfrak{M}$ equipped with
the atlas described above. We call it the \emph{direct union} of our
regular direct system of smooth manifolds. 
\end{defn}

\begin{rmk}
None of the operations involved in the construction above requires any
notion of choice.
\end{rmk}

From the considerations above we see that $f_{*i}$ is continuous,
open, and injective. So it is a homeomorphism onto its image.
Furthermore, it is by our choice of smooth structure smooth. Therefore
we have that

\begin{prop}
The mappings $f_{*i}: M_i \to \dirlim \mathfrak{M}$ are open and
diffeomorphic onto their image. 
\end{prop}

\section{Direct limit maps}
Now suppose $(\mathfrak{M},\mathfrak{F})$ and
$(\mathfrak{N},\mathfrak{G})$ are two regular direct systems of smooth
manifolds indexed by the same direct set $(I,\prec)$. Suppose
furthermore that there is a set $\mathfrak{H} = \{h_i\}_{i\in I}$
where $h_i: M_i \to N_i$ are smooth maps such that whenever $i\prec j$
\begin{equation}\label{eq:hcompatibility} h_j \circ f_{ji} = g_{ji} \circ h_i~.\end{equation}

\begin{prop}\label{prop:limitmap}
There exists a smooth map $h_*: \dirlim \mathfrak{M} \to \dirlim
\mathfrak{N}$ such that for every $i$, 
\[ h_* \circ f_{*i} = g_{*i} \circ h_i~.\]
\end{prop}

\begin{proof}
It suffices to check that $h_*$ is well-defined; for this we only need
to check that if $x_i\sim x_j$ in $\coprod_i M_i$ that $h_i(x_i) \sim
h_j(x_j)$ in $\coprod_j N_j$. But this follows from
\eqref{eq:hcompatibility}. That $h_*$ is a smooth map follows from the
smoothness of $h_i$ and the fact that if $(U_i,\phi)$ is a chart for
$M_i$ and $(V_j,\psi)$ for $N_j$
\[ \psi\circ g_{*j}^{-1} \circ h_* \circ f_{*i} \circ \phi^{-1} = \psi
\circ g_{kj}^{-1} \circ h_k \circ f_{kj} \circ \phi^{-1} \]
for any $k\succ i,j$ over any open set where all the operations are
defined, and that an atlas for the direct union manifold is given by
the collection of all pushforward charts. 
\end{proof}

A direct consequence of the above construction is that we can take the
direct union of a regular direct system of fibred manifolds (under the
obvious definition); similarly, if this system is equipped with smooth
sections that obey an appropriate commutation relation of the form
\eqref{eq:hcompatibility}, we can extend this section to a section
over the direct union of the base manifolds. In particular, noting
that a pseudo-Riemannian metric on a smooth manifold $M$ 
is a section (always smooth; see Remark \ref{rmk:smoothness}) of the vector bundle $T^{0,2}M$, we have

\begin{cor}
Let $(\mathfrak{M},\mathfrak{F})$ and $(\mathfrak{N},\mathfrak{G})$ be
two regular direct systems of smooth manifolds over the same direct
set $(I,\prec)$, and assume that $M_i$
and $N_i$ are equipped with pseudo-Riemannian metrics such that the
mappings $f_{ji}$ and $g_{ji}$ are isometries onto their image. If 
furthermore we have
a set $\mathfrak{H} = \{h_i\}_{i\in I}$ of isometries $h_i: M_i \to
N_i$, such that \eqref{eq:hcompatibility} is satisfied. Then
\begin{enumerate}
\item $\dirlim\mathfrak{M}$ and $\dirlim\mathfrak{N}$ can be equipped
with pseudo-Riemannian metrics such that the mappings $f_{*i}$ and
$g_{*i}$ are isometries. 
\item The smooth map $h_*$ of Proposition \ref{prop:limitmap} is an
isometry. 
\end{enumerate}
\end{cor}

\begin{rmk}\label{rmk:smoothness}
In terms of applications to physics (see Section
\ref{sec:app}), We only consider the case of
\emph{smooth} solutions to the Cauchy problem. Most of the statements
here carry over exactly when the relevant structures are of class
$C^k$ which are suitably compatible under compositions. For some of
the difficulties involved when considering Sobolev-class structures
and in developing the \emph{black box} local existence and uniqueness
results (cf.\ Section \ref{sec:app}) in lower (but still classical) 
regularity, see
\cite{ChruscP2011}.
\end{rmk}

\section{Separation and countability}
For reasons of analysis, it is usually convenient to work with smooth
manifolds that are Hausdorff and second countable. Let us first
consider the separation axioms. We have the following
\begin{lem}
Let $(\mathfrak{M},\mathfrak{F})$ be a regular direct system of smooth
manifolds. Let $x_i, y_i \in M_i$ and let $U_i \subseteq M_i$ be open,
such that $x_i \in U_i$ and $y_i \not\in U_i$. Then $f_{*i}(x_i) \in
f_{*i}(U_i)$, $f_{*i}(y_i) \not\in f_{*i}(U_i)$, and $f_{*i}(U_i)$ is
open. 
\end{lem}
\begin{proof}
By the construction in Section \ref{sec:directunion}, $f_{*i}$ is open
and injective; the lemma follows. 
\end{proof}

Now recall some of the separation axioms. A topological space is said to be
\begin{itemize}
\item $\mathsf{T}_0$ (Kolmogorov) if given any $x\neq y$, there exists an open set $U$
such that exactly one of $x,y$ belongs to $U$.
\item $\mathsf{T}_1$ (Fr\'echet) if given any $x\neq y$, there exists 
open sets $U,V$ such that $x\in U$, $y\in V$ and $y\not\in U$, $x\not\in V$. 
\item $\mathsf{T}_2$ (Hausdorff) if given any $x\neq y$, there
exists open sets $U,V$ such that $x\in U$, $y\in V$ and $U\cap V =
\emptyset$. 
\end{itemize}
\begin{cor}
Let $(\mathfrak{M},\mathfrak{F})$ be a regular direct system of smooth
manifolds. Assume the elements $M_i$ are all $\mathsf{T}_0$ (resp.\
$\mathsf{T}_1$ or
$\mathsf{T}_2$) as topological spaces. Then $\dirlim_{\text{Top}}\mathfrak{M}$
is $\mathsf{T}_0$ (resp.\ $\mathsf{T}_1$ or $\mathsf{T}_2$). 
\end{cor}

Countability, on the other hand, is problematic. 
\begin{exa}
Let $\mathfrak{M}$ consists of \emph{finite} subsets of $(0,1)$,
equipped with the discrete topology; the elements are trivially
$0$-dimensional smooth manifolds. By definition each $M_i$ is second
countable, separable, and Lindel\"of. But the (direct) union $\dirlim
\mathfrak{M}$, which is again $(0,1)$ with the discrete topology, is
none of the three. 
\end{exa}
One may be tempted into thinking that the issue can be fixed by
working with spaces that are connected. But then one runs into the
problem with the long line (which can be defined as a direct limit of
uncountably many copies of $\Real$). 

We return to the resolution of this problem when we discuss physical
applications in Section \ref{sec:app}.

\section{A note on the inverse limit}
Just as the direct limit generalises the notion of \emph{unions} of
sets, we can use the notion of the \emph{inverse limit} to generalise
the notion of \emph{intersections}. More precisely, given a direct
system of smooth manifolds $(\mathfrak{M},\mathfrak{F})$, we consider the
set
\[ \invlim_{\text{Set}} \mathfrak{M} = \set{\vec{x} \in \prod M_i}{x_j
= f_{ji}(x_i) \text{ whenever } i\prec j}~.\]
A priori the inverse limit can be empty, since the intersection of an
arbitrary family of sets can be the empty set. But moreover, to even
assert that $\prod M_i$ is non-empty in general is precisely the
\emph{axiom of choice}. 

Ignoring this problem with choice, we also see that in many regards
the inverse limit does not behave as nicely as the direct limit
(union). For example, the inverse limit of a directed system of open sets
is not necessarily open: consider $\mathfrak{M} =
\set{(-q,q)\subset\Real }{q\in \Rational\cap (-1,1)}$ ordered by
inclusion, with $\mathfrak{F}$ the inclusion maps. Their direct union
is their union which is $(-1,1)$, but their inverse limit is the
single point $\{0\}$. Given a regular direct system of smooth
manifolds, supposing that the inverse limit exists, the most we can
say is that the projection map $f_{i*}: \invlim_{\text{Set}}
\mathfrak{M}\to M_i$ is injective. The projection maps are not
guaranteed to be open, nor can we easily define a smooth structure. 

Luckily, in the context of the evolution problem in physics, we need
not consider the inverse limit, as the appropriate object is already
given to us \emph{as the initial data}. 

\section{Application to physics}\label{sec:app}
In physics, as motivated in the introduction, we want to consider the
initial value problem to some systems of equations. The simplest
initial value problem is that of an ordinary differential equation
(ODE).
Here we immediately see that the issue of the long line cannot arise.
This is due to the demand that we have an increasing time function on
our ``solution manifold'', which prevents the interval on which the
solution exists from getting too long. We can codify this intuition by
requiring that there be a well-defined time-function for test
particles. 

\begin{defn}\label{defn:physical}
A regular directed system of smooth manifolds
$(\mathfrak{M},\mathfrak{F})$ is said to be \emph{physical} if there
exists a smooth manifold $\Sigma$ and a family $\mathfrak{H} =
\{h_i\}_{i\in I}$ of injective continuous open maps $h_i: M_i \to \Sigma\times \Real$
such that $h_j\circ f_{ji} = h_i$. 
\end{defn}
As an immediate consequence of the definition, there exists an
injective continuous open map $h_*: \dirlim \mathfrak{M} \to
\Sigma\times\Real$, which implies that $h_*$ is an homeomorphism onto
its images. Therefore if $\Sigma$ is second-countable, we will also
have that $\dirlim\mathfrak{M}$ is second countable. 

The definition above captures crucially a notion of \emph{global
hyperbolicity} of solutions to initial value problems. Recall that a 
consequence of
global hyperbolicity for smooth Lorentzian manifolds is that
\emph{every inextensible time-like geodesic must intersect the Cauchy
hypersurface exactly once}. This in particular means (see also
\citet[p.180]{Ringst2009} for a similar argument)
\begin{lem}\label{lem:count}
Let $(\tilde{M},g)$ be a globally hyperbolic Lorentzian manifold, and
let $\tilde{\Sigma}$ be a Cauchy hypersurface. Then $\tilde{\Sigma}$
being second-countable implies $\tilde{M}$ is second-countable. 
\end{lem}
\begin{proof}
Define $M = \set{ x = (p,v)\in T\tilde{M}}{g(v,v) = -1}$ and $\Sigma =
T\tilde{\Sigma}$. If $\tilde{\Sigma}$ is second countable, so is
$\Sigma$. By global hyperbolicity, corresponding to each $x\in M$
there is exactly one $y = (\sigma,w)\in \Sigma$ and $s\in\Real$ 
such that the geodesic $\gamma$ in $\tilde{M}$ with initial value $x$
has $\gamma(s) = \sigma$ and $\dot{\gamma}(s)$ projects orthogonally
to $w$. The mapping $x\mapsto y$ is continuous and injective by the
wellposedness theory of ODEs. Reversing the flow we also have that the
mapping is open. Thus $M$ is homeomorphic to an open subset of $\Sigma
\times \Real$. This implies that $M$ is second countable, and since
the bundle projection $M\to \tilde{M}$ is open, so is $\tilde{M}$. 
\end{proof}
\begin{rmk}
The usual proof that $\tilde{\Sigma}$ is second countable implies
$\Sigma$ is second countable uses the statement ``countable union of
countable sets is countable'', which is a weak form of choice. If we
assume however the initial data is specified such that $\Sigma =
T\tilde{\Sigma}$ is a second countable manifold, this deficiency can be
circumvented. Note that the proof that finite Cartesian products of second
countable manifolds is second countable follows from the fact that
Cartesian products of countable sets are countable, which does
\emph{not} require any form of choice by the Cantor argument. 
\end{rmk}

In general, one can think of Definition \ref{defn:physical} as
requiring there be 
\begin{enumerate}
\item a well defined evolution for ``test particles'';
\item and a notion of hyperbolicity which states that every test
particle in the space-times can be traced back to one that arose from
some initial data. 
\end{enumerate}
The $\Sigma$ factor of the mapping $h_i$ gives the initial
configuration in phase space of the test particle, and the $\Real$
factor gives the elapsed (proper) time.

To apply the general machinery we have developed above to an \emph{initial
value problem}, we require that the solutions to the initial value
problem satisfy certain nice properties. Below is the general prescription:

We represent the initial data of the problem by some fibred manifold
$\Sigma$. By a set of solutions  we refer to a set $\mathfrak{M}$ of
fibred manifolds such that for each $M\in \mathfrak{M}$, there is an
embedding $\phi: \Sigma \to M$. 
\begin{rmk}
We have to be very careful
here when speaking of the ``set'' of solutions. Recall that Einstein's
equation is diffeomorphism invariant. Now given a set of diffeomorphic
manifolds $\mathfrak{M}$, we can apply the direct union construction to get a new
manifold $\dirlim\mathfrak{M}$ which is also diffeomorphic to
any element of $\mathfrak{M}$. By the construction of
$\dirlim_{\text{Set}}\mathfrak{M}$ as an equivalence class,
we must have that $\dirlim\mathfrak{M}\not\in \mathfrak{M}$. To put it in
another way, by the axiom of regularity it does not make sense to
speak of ``the set of all manifolds diffeomorphic to $M$'', and so
it also does not make sense to speak of the set of all manifolds
solving Einstein's equations with a given initial value. It is for
this reason that in the \emph{locality} part of Definition
\ref{defn:ivp} below we cannot directly require that $M \in
\mathfrak{M}$. 
\end{rmk}
Given $M,M'\in \mathfrak{M}$, and
$\phi,\phi'$ their corresponding embeddings of $\Sigma$, we say
that $M$ is an \emph{extension} of $M'$ if there exists an open
embedding $f: M' \to M$ such that $f\circ \phi' = \phi$. Denote by
$\mathfrak{F}$ a set of extension maps.
\begin{defn}\label{defn:ivp}
Given the pair $(\mathfrak{M},\mathfrak{F})$ of a set of solutions to
the initial value problem and a set of extension maps, we say that it
satisfies
\begin{itemize}
\item \emph{existence} if $\mathfrak{M}$ is non-empty;
\item \emph{unique extension} if $f,f'\in \mathfrak{F}$ are both
extension maps sending $M' \to M$, then $f = f'$;
\item \emph{locality} if ``unions of solutions is a solution''; by
this we mean if $M$ is a smooth fibred manifold such that there exists open
embeddings $f_i: M_i \to M$ where $M_i\in \mathfrak{M}$, such that
\begin{enumerate}
\item $M = \cup_i f_i(M_i)$
\item $f_i$ commute with the embedding
of the initial data $\Sigma$ 
\item if for some $i,j$ there exists $M'\in \mathfrak{M}$ with $f'_i: M' \to M_i$ and
$f'_j: M' \to M_j$, then $f_i \circ f'_i = f_j\circ f'_j$ (preserves
unique extension) 
\end{enumerate} 
then there exists $N\in \mathfrak{M}$
and a diffeomorphism $g: M\to N$ such that $g\circ f_i \in
\mathfrak{F}$ is the extension map from $M_i$ to $N$.
\item \emph{uniqueness} if $M,M'\in \mathfrak{M}$, there exists $N\in
\mathfrak{M}$ such that $M,M'$ are both extensions of $N$;
\end{itemize}
\end{defn}
\begin{rmk}\label{rmk:properties}
For Einstein's equations, existence here is the same as \citet[Theorem 1]{ChoGer1969}, and
uniqueness is the same as \citet[Theorem 2]{ChoGer1969}. The locality
property holds for any initial value problem expressible as finite
order partial differential equations (and in particular Einstein's
equations); but additional properties like Hausdorff separation may
need checking. The unique extension property is implicit in
\citet{ChoGer1969}, even though it is not explicitly mentioned. 
\end{rmk}

\begin{rmk}\label{rmk:setsoln}
The author should point out that, for Einstein's
equation, the fact that there exists a set (and not a proper class) $\mathfrak{M}$
such that for every globally hyperbolic Cauchy development $N$ there
exists an elements $M\in \mathfrak{M}$ such that $N$ is diffeomorphic
to $M$, depends strongly on the requirements that the solutions are 
Hausdorff and second countable. Dropping these requirements, given 
a local solution $M$ of
Einstein's equation with the given data, for every ordinal
$\lambda$, we can form the Cartesian product $\lambda \times M$. This
shows that in general the collection of \emph{all} solutions of
Einstein's equation must form a proper class. (This example is due to 
Tobias Fritz \citeyearpar{Fritz2013}.) This
implicit assumption where one is working with a set of solutions and
not a larger collection is not justified by \citet{ChoGer1969}, and is used when it is asserted that
there is a poset structure on the collection of solutions. Here we
give a sketch (this proof is inspired by some comments of Omar
Antol\'\i{}n-Camarena and Igor Belegradek on MathOverflow; see URL
given in \cite{Fritz2013}):

Assuming solutions are Hausdorff and second countable and connected (a
consequence of global hyperbolicity), a solution is representable by
sections of some finite dimensional fibred manifold. Appealing to the
strong Whitney embedding theorem (the author makes no claim on its
dependence on some version of axiom of choice) the solution is diffeomorphic to some
smooth submanifold of $\Real^K$ for some $K$ sufficiently large.
Therefore there exists a subset of the power set
$\mathcal{P}(\Real^K)$ which contains the diffeomorphic image of any 
solution. This implies that we have an appropriate ``set of
solutions''. 

We further remark that this point is also explicitly considered by
\citet{Ringst2009}, who gave a different proof then the sketch below,
making use of the Geroch splitting theorem (see \citet[Property 7,
p.444]{Geroch1970} and \citet{BerSan2003}); a similar approach is taken
by \citet{Sbi2013}.
\end{rmk}

\begin{thm}
Given the set $\mathfrak{M}$ and $\mathfrak{F}$ corresponding to an
initial value problem that satisfies existence, uniqueness, unique
extension, and locality, then there exists a \emph{maximal} solution;
that is to say, there exists $\bar{M}\in \mathfrak{M}$ such that for
any $M\in \mathfrak{M}$ there exists an open embedding $f\in
\mathfrak{F}$ sending $M\to \bar{M}$. 
\end{thm}
\begin{proof}
We divide into several steps. 

\textbf{Construction of overlaps}. Given $M, M'\in \mathfrak{M}$, the
uniqueness property guarantees that the set 
\begin{multline*} C^{-}(M,M') \eqdef \{N \in
\mathfrak{M} \,|\, \exists f_N,f'_N\in \mathfrak{F}, \\f_N:N\to M, f'_N:N\to
M'\} \end{multline*}
is non-empty. Let 
\[ \bar{N} = \cup_{N\in C^-(M,M')} f_N(N) \]
and
\[ \bar{N}' = \cup_{N\in C^-(M,M')} f'_N(N)~.\]
Since $f_N$ and $f'_N$ are in $\mathfrak{F}$, unique extension is
verified and so by locality we have that there exists 
$\tilde{N}, \tilde{N}'\in
\mathfrak{M}$ diffeomorphic to $\bar{N}$ and $\bar{N}'$. Now, using
that $f'_N \circ f_N^{-1}|_{f_N(N)}$ is a diffeomorphism onto its
image, by Proposition \ref{prop:limitmap} there exists a
diffeomorphism $\tilde{f}: \tilde{N} \to \tilde{N}'$, and hence
$\tilde{N},\tilde{N}'\in C^-(M,M')$. In other words, there exists some
maximal element of $C^-(M,M')$ in the sense that it is an extension of 
any other element of $C^-(M,M')$. We denote the solution $\tilde{N}$
constructed above by $M\wedge M'$ and call it the \emph{maximal
overlap}. 

\textbf{Construction of pairwise extensions}. Now given $M,M'\in
\mathfrak{M}$, let $f,f'$ be the extension map from $M\wedge M'$ to
$M$ and $M'$ respectively. Let us consider $N = M\coprod M' / \sim$
where the equivalence relation is that $x \sim x'$ iff there exists
$y\in M\wedge M'$ such that $f(y) = x$ and $f'(y) = x'$. By the
maximality of $M\wedge M'$ we have that unique extension into $N$ is
preserved, so by locality there exists a solution $M\vee M'$
diffeomorphic to $N$.  

\textbf{Building the maximal element}. By the previous steps, we have
that the set $\mathfrak{M}$ in fact forms a lattice with the partial
ordering induced by the extension maps. In particular, it is in itself
a directed set. Hence we can take the direct union $\dirlim
\mathfrak{M}$. The direct union construction guarantees that unique
extension is preserved. Therefore we can once again appeal to locality
to conclude that there must exists some element $\bar{M}$ of 
$\mathfrak{M}$ that is diffeomorphic to $\dirlim \mathfrak{M}$. 
\end{proof}

\begin{cor}
For Einstein's equations in general relativity, there exists a
Hausdorff, second countable, maximal globally hyperbolic Cauchy
development. 
\end{cor}
\begin{proof}[Sketch of proof]
As indicated in Remark \ref{rmk:setsoln}, assuming our initial data
$\Sigma$ is Hausdorff and second countable, the set of ``Hausdorff,
second countable, globally hyperbolic Cauchy developments'' forms a
set $\mathfrak{M}$. And as indicated in Remark \ref{rmk:properties},
that this set is nonempty (\emph{existence}), and two developments must
overlap (\emph{uniqueness}) is classical. As part of the uniqueness
statement it is implicit that that \emph{unique extension} also holds
for the Einstein system. It remains to show that \emph{locality}
holds. The main difficult step occurs here, in showing that pairwise
extensions remain Hausdorff. The proof occupies most of section 3.2 in
Sbierski's pre-print (while being briefly sketched in the original
paper of Choquet-Bruhat and Geroch), and we omit it here. For second
countability a simple application of Lemma \ref{lem:count} suffices. 
\end{proof}

\section{Acknowledgements}
The author wishes to thank Asaf Karagila for his enlightening
explanations of aspects of set theory, Andrew Stacey for clarifying
some subtleties in global analysis, and Mihalis Dafermos and Jan
Sbierski for comments on the manuscript. 
The author especially wants to acknowledge Jan Sbierski for
making available a copy of his delightful manuscript prior to its
appearance on \url{http://arXiv.org}.

The author is supported by the Swiss National Science
Foundation. He would like to also thank the Mathematical Sciences
Research Institute for their hospitality during his stay as a part of
the 2013 program in mathematical general relativity; the research
leading to this manuscript was initiated at the MSRI.

\end{document}